\documentclass[letterpaper]{article} 
\usepackage{aaai2026} 
\nocopyright
\usepackage{xr}
\usepackage{times} 
\usepackage{helvet} 
\usepackage{courier} 
\usepackage[hyphens]{url} 
\usepackage{graphicx} 
\usepackage{natbib} 
\usepackage{caption} 
\usepackage{amsmath}
\usepackage{amssymb}
\usepackage{amsfonts}
\usepackage{amsthm}
\usepackage{array}
\usepackage{booktabs}
\usepackage{multirow}
\usepackage{float}
\usepackage{placeins}
\usepackage{xcolor}

\newtheorem{theorem}{Certificate}
\newtheorem{lemma}{Lemma}
\newtheorem{proposition}{Proposition}
\newtheorem{definition}{Definition}
\newtheorem{assumption}{Assumption}

\newcolumntype{L}[1]{>{\raggedright\arraybackslash}p{#1}}

\title{Certified Speculative Execution for\\
Untrusted AI Agents}
\author{
    Chenyu Zhou\textsuperscript{\rm 1}\textsuperscript{*\dag},
    Qiliang Jiang\textsuperscript{\rm 2}\textsuperscript{*},
    Shuning Wu\textsuperscript{\rm 3},
    Xu Zhou\textsuperscript{\rm 3}\textsuperscript{\dag}
}
\affiliations{
    \textsuperscript{\rm 1}School of Engineering, Institute of Science Tokyo, Japan\\
    \textsuperscript{\rm 2}College of Control Science and Engineering, Zhejiang University, China\\
    \textsuperscript{\rm 3}Department of Electrical and Computer Engineering, National University of Singapore, Singapore\\
    zhou.c.76d6@m.isct.ac.jp, jiangqiliang@zju.edu.cn, shuningwu@u.nus.edu, zhouxu\_nus@u.nus.edu\\[0.4ex]
    \textsuperscript{*}Equal contribution.\quad\textsuperscript{\dag}Corresponding authors.
}

\begin{document}
\maketitle

\begin{abstract}
Hard-constrained sequential decision systems have no certified way to
spend the test-time compute of modern AI: executing the multi-step drafts
of a learned policy or a frozen LLM forfeits the feasibility guarantee a
trusted solver provides, while invoking the solver at every step forfeits
the speed the AI offers. Certificate-Gated Prefix Acceptance (CGPA) closes
this gap with a \emph{certified speculative-execution contract} for
untrusted AI agents: a trusted
verifier rejects constraint-violating transitions exactly, a conformally
calibrated value boundary gates the longest low-cost prefix within a
per-segment regret budget, and the rest defers to the solver --- so
safety, regret, and speed decouple by construction. The contract drives
every untrusted proposal source --- adversarial drafters and six
heterogeneous frozen LLMs (including a $12$B model that violates
constraints in $98\%$ of direct rollouts) --- to zero applied
violations; a certificate-aware learned boundary, conformally calibrated,
drives mean regret three orders of magnitude below unguarded acceptance, to
within sampling noise of the stepwise oracle ($95\%$ CI spanning zero), and
under calendar shift a learned proposal source overtakes it on $15$ of
$18$ held-out days. On a deployment-scale unit-commitment
instance it turns a frozen $8$B LLM into a $2.96\times$ per-episode wall-clock
speedup at $2.1\%$ regret, outpacing the domain heuristic ($1.79\times$) and a safe
receding-horizon baseline ($1.07\times$): the more
capable the untrusted source, the faster the certified system, at
guarantees that never change.
\end{abstract}


\section{Introduction}

Modern AI can draft sequences of decisions at low computational cost --- a learned policy
rolling out actions, a frozen large language model proposing a plan. Acting
on those drafts directly in \emph{hard-constrained} settings, however,
forfeits safety: a learned or LLM policy carries no feasibility guarantee
--- a frozen $12$B-parameter LLM proposes
constraint-violating battery actions in $98\%$ of episodes, and an
end-to-end reinforcement learner replacing the trusted module pays about $7\times$ its regret. The reliable alternative delegates every step to an
expensive trusted solver --- an LP/MILP optimizer, a model-predictive
planner, a shield --- that certifies feasibility and supplies a safe
repair on demand. But per-step
delegation pins the most expensive
component to the critical path even when short-horizon action sequences
are largely predictable. Safely deploying untrusted AI therefore needs a
certified layer that \emph{uses} AI drafts without \emph{trusting} them
--- one in which a thin trusted core owns every guarantee.

The missing layer has a precise shape: intelligence and guarantees must
come from different owners. Drafts must come from a source that is free to
be arbitrarily capable --- and arbitrarily wrong --- while every guarantee
is owned by a thin trusted core that never depends on the source. This
separation has one mature precedent: speculative decoding, where a cheap
drafter proposes tokens and the target model verifies a prefix and accepts
it in one pass \cite{leviathan2023speculative,chen2023speculative}. But a
token draft is checked against a probability distribution; a control draft
must be checked against \emph{hard feasibility} and against \emph{decision
quality}, and these two checks cannot be the same object, because a
constraint-safe action prefix can still be economically poor. What
verification means for executed decisions --- and what acceptance is worth
--- has to be rebuilt from the ground up.

We make this precise as Certificate-Gated Prefix Acceptance (CGPA), a single
\emph{certified speculative-execution contract} for untrusted AI agents. An
untrusted proposal source --- a rule, a learned policy, or a frozen LLM ---
proposes a length-$K$ action prefix; a trusted verifier rejects any unsafe
transition, a trusted value boundary truncates prefixes that are safe but
costly, and CGPA commits the longest certified prefix and \emph{defers} the
rest to the oracle.

This contract's three parts are structurally necessary: remove any one and it
collapses onto a known weaker baseline
(Section~\ref{sec:method}). Any layer that turns
untrusted test-time compute into safe solver amortization must own
feasibility exactly, or the proposal source inherits a path to violations
(Proposition~\ref{prop:decoupling}); must budget excess cost per accepted
segment, or acceptance has no auditable price
(Certificate~\ref{thm:approx}); and must calibrate the one quantity it
estimates, or that budget is an assumption rather than a measurement
(Proposition~\ref{prop:conformal-bound}).

The contract is one object --- a certified prefix-acceptance operator ---
whose regret certificate charges the value-boundary error $\epsilon$
\emph{once per accepted segment}, not per step (Section~\ref{sec:method}).
This single accounting identity organizes the paper: safety never queries
$\epsilon$ (Section~\ref{sec:safety}); learning $\epsilon$ down shrinks regret
three orders of magnitude (Section~\ref{sec:quality}); a heterogeneous source
moves only which prefixes are accepted (Section~\ref{sec:diversity}); and each
accepted segment skips its oracle solves, yielding a $2.96\times$ wall-clock
speedup on deployment-scale unit commitment (Section~\ref{sec:deploy}).

\section{Certified Speculative Execution}
\label{sec:method}

\begin{figure*}[t]
\centering
\includegraphics[width=\textwidth]{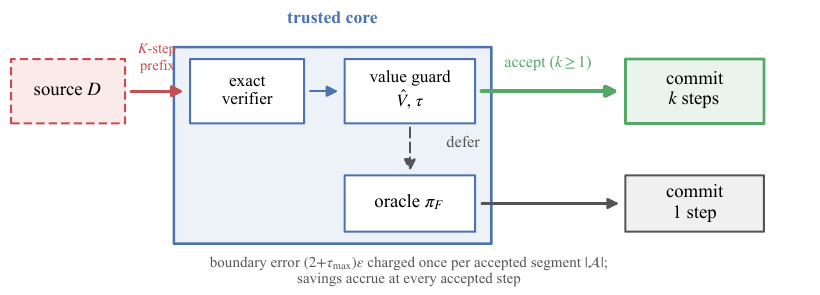}
\caption{The certified prefix-acceptance operator
(Definition~\ref{def:operator}). An untrusted source drafts a $K$-step
prefix; a thin trusted core (exact verifier, certificate-aware value
boundary, fallback oracle) commits the longest certified prefix and defers
the rest. Safety never queries the source; boundary error is charged once per
accepted segment while savings accrue at every accepted step.}
\label{fig:operator}
\end{figure*}

\subsection{Problem Setup}
We consider finite-horizon decision control with state $s_t$, action
$a_t$, transition $s_{t+1}=f(s_t,a_t)$, non-negative stage cost
$c(s_t,a_t)\ge0$, and hard feasibility constraints. The remaining horizon
is encoded in $s_t$, so the oracle cost-to-go $V(s_t)$ is written without a
time subscript. A trusted oracle policy $\pi_F$ returns feasible actions; the
reference baseline invokes $\pi_F$ at every step and applies only the first
action of each receding-horizon oracle solve, defining
\begin{equation}
J_F(s_0) = \sum_{t=0}^{T-1} c(s_t,\pi_F(s_t)).
\end{equation}
For a CGPA rollout with applied actions $a_t^{\mathrm{CGPA}}$,
$J_{\mathrm{CGPA}}(s_0)=\sum_{t}c(s_t,a_t^{\mathrm{CGPA}})$, and we report
the fallback-relative normalized regret used in every table,
\begin{equation}
R_{\mathrm{rel}} =
\frac{J_{\mathrm{CGPA}}(s_0)-J_F(s_0)}
{\max(|J_F(s_0)|,10^{-12})}.
\label{eq:regret}
\end{equation}
CGPA preserves the oracle's safety within an auditable regret budget while
moving it off the per-step critical path.

\subsection{The Accept/Defer Contract}
\label{sec:cgpa}

\begin{definition}[Certified prefix-acceptance operator]
\label{def:operator}
A CGPA controller is a tuple $(D,\mathrm{Verify},\hat V,\pi_F)$ --- an
untrusted proposal source, a trusted per-transition feasibility check, a
lightweight value boundary, and a trusted oracle policy --- read as a single
\emph{certified prefix-acceptance operator} $\mathcal{C}$ over the trusted core
$(\mathrm{Verify},\hat V,\pi_F)$, parameterized by the untrusted source $D$
(Figure~\ref{fig:operator}): at each state it
commits the longest verified-safe proposed prefix admitted by the value
guard \eqref{eq:value_guard}, and otherwise defers to a single verified
oracle repair. Only $\mathrm{Verify},\hat V,\pi_F$ are trusted; the proposal
source $D$ is \emph{untrusted and may be adversarial}, and every guarantee
below depends only on $\mathrm{Verify},\hat V,\pi_F$, never on $D$.
$\mathcal{C}$ maintains two invariants at every executed step:
(\textbf{I1})~\emph{viability}, every executed transition lies in the
verifier-admissible set; and (\textbf{I2})~\emph{value descent}, every
accepted segment obeys the local descent condition \eqref{eq:dissipation}
on $\hat V$. The proposal source occupies the untrusted slot and enters
neither invariant.
\end{definition}

At state $s_t$, the proposal source proposes $K$ actions
$u_{0:K-1}=D(s_t,K)$. A trusted verifier simulates the prefix and stops at
the first unsafe transition. Among the verified-safe prefix lengths, CGPA
selects the longest prefix satisfying
\begin{equation}
g(s_t,u_{0:k-1})+\hat V(s_{t+k})
\leq \hat V(s_t)+\tau |\hat V(s_t)| ,
\label{eq:value_guard}
\end{equation}
where $g(s_t,u_{0:k-1})$ is the verified prefix cost and $\tau$ is the
value tolerance (the absolute value enables scaling with the boundary
magnitude). If $k>0$ the controller executes the accepted prefix; if $k=0$
it defers, applying one verified oracle repair.

Equation~\eqref{eq:value_guard} is the boundary-estimated form of a
\emph{per-segment regret budget}: it asks that the realized prefix cost not
exceed the drop in cost-to-go $\hat V(s_t)-\hat V(s_{t+k})$ by more than a
fraction $\tau$ of the local value scale --- exactly the condition that
\emph{telescopes} into a bounded episode-level certificate
(Section~\ref{sec:certificate}). For non-negative prefix cost it is, after
rearrangement, a local \emph{value-descent condition}
\begin{equation}
\hat V(s_{t+k}) - \hat V(s_t)\;\le\; -\,g(s_t,u_{0:k-1}) + \sigma_t,
\quad \sigma_t := \tau\,|\hat V(s_t)| ,
\label{eq:dissipation}
\end{equation}
read in the cost-to-go convention $\hat V\ge 0$ (so $-g\le0$,
$\sigma_t\ge0$).

The value boundary may be precomputed, learned, or forecast-based, provided its online query cost is small relative to the
oracle solve. The guarantees rest on three assumptions separating trusted
from untrusted components.

\begin{assumption}[Trusted oracle]
\label{asm:fallback}
At every reachable state $s$, $\pi_F$ returns a feasible action that the
verifier admits. A safe repair is always available.
\end{assumption}

\begin{assumption}[Bounded boundary error]
\label{asm:bounded}
The value boundary satisfies $|\hat V(s)-V(s)|\le\epsilon$ at the endpoint
states of accepted segments. This is a reachable-set bound we audit
empirically over the accepted-segment endpoints reached under CGPA
(Section~\ref{sec:learned}); the
exact-boundary results take $\epsilon=0$.
\end{assumption}

\begin{assumption}[Arbitrary proposal source]
\label{asm:drafter}
The proposal source $D$ is unconstrained and may be rule-based,
forecast-based, learned, or LLM-based. CGPA makes no safety or optimality
assumption on $D$; every guarantee derives from $\mathrm{Verify}$,
$\hat V$, and $\pi_F$.
\end{assumption}

\subsection{Safety Holds for Any Proposal Source}

CGPA never applies a proposed transition the trusted verifier has not
accepted. When the verifier rejects at length zero, the oracle supplies a
repair, which is itself verified before execution.

\begin{proposition}[Applied safety]
\label{prop:safety}
Under Assumption~\ref{asm:fallback}, every transition applied by CGPA is
feasible, for any proposal source (Assumption~\ref{asm:drafter}).
\end{proposition}

\begin{proof}
A proposed action is applied only if the verifier admits it; on rejection at
length zero, the controller applies an oracle repair, which
Assumption~\ref{asm:fallback} guarantees exists and is verified first. No
unverified transition is ever applied, regardless of $D$.
\end{proof}

The next proposition shows the decoupling is necessary, not stylistic: the
feasibility gate cannot be folded into the learned value boundary.

\begin{proposition}[Exact feasibility gating]
\label{prop:decoupling}
If feasibility is decided by an exact, oracle-owned gate, then applied
safety (Proposition~\ref{prop:safety}) holds for every proposal source and
every boundary error $\epsilon$, because that argument never queries
$\hat V$. If feasibility is instead decided by a non-conservative
$\epsilon$-accurate surrogate, then applied safety need not be independent
of the boundary error: take a state with feasible $u_1$ and infeasible
$u_2$ thresholded at level $\theta$; if the true score of $u_2$ is
$\theta-\delta$ with $0<\delta<\epsilon$, the surrogate may admit $u_2$
while an exact test rejects it for every $\epsilon$. Hence the learned
value boundary cannot double as the feasibility gate.
\end{proposition}

\subsection{Regret Certificate}
\label{sec:certificate}

Let $V(s)=V^{\pi_F}(s)$ be the finite-horizon oracle cost-to-go, terminal
value zero.

\begin{definition}[Accepted and repair segments]
\label{def:segments}
A CGPA rollout partitions into consecutive segments whose end state is the
next start state. An \emph{accepted segment} $i$ starts at $s_i$, executes
$k_i\ge1$ verified actions admitted by \eqref{eq:value_guard}, incurs cost
$g_i$, ends at $s_{i+1}$, and uses tolerance $\tau_i$ (all reported
experiments fix $\tau_i=\tau$; the bound is stated for general $\tau_i$). A
\emph{repair segment} $j$ is a single oracle action applied on deferral,
with cost $f_j$ from $s_j$ to $s_{j+1}$. The accepted and repair sets are
$\mathcal A$ and $\mathcal R$.
\end{definition}

For an exact value $V$, \eqref{eq:value_guard} specializes to admitting a
segment when $g_i + V(s_{i+1}) \leq V(s_i)+\tau_i |V(s_i)|$.

\begin{lemma}[Per-segment value decomposition]
\label{lem:segment}
Each accepted segment obeys
$g_i \leq V(s_i)-V(s_{i+1})+\tau_i|V(s_i)|$; each repair segment obeys
$f_j \leq V(s_j)-V(s_{j+1})+\rho_j$, where $\rho_j\ge0$ is the realized
repair slack defined by $f_j+V(s_{j+1})\leq V(s_j)+\rho_j$.
\end{lemma}

\begin{theorem}[Exact-boundary fallback-relative bound]
\label{thm:certificate}
Under Assumption~\ref{asm:fallback}, if every accepted segment satisfies the
exact guard, the episode excess cost satisfies
\begin{equation}
J_{\mathrm{CGPA}}(s_0)-J_F(s_0)
\leq \sum_{i\in\mathcal A}\tau_i |V(s_i)|
 + \sum_{j\in\mathcal R}\rho_j .
\end{equation}
\end{theorem}

Summing Lemma~\ref{lem:segment} over all segments telescopes to
$V(s_0)$; dividing by $|J_F|$ recovers the normalized
regret~\eqref{eq:regret} (full proof in the appendix).
The deployed controller runs \eqref{eq:value_guard} with
$\hat V$, not $V$. The next result propagates the boundary error through the
same argument.

\begin{theorem}[Approximate-boundary fallback-relative bound]
\label{thm:approx}
Under Assumptions~\ref{asm:fallback} and~\ref{asm:bounded} with
$\tau_{\max}=\max_i\tau_i$, if every accepted segment satisfies the deployed
rule \eqref{eq:value_guard}, then
\begin{equation}
\begin{aligned}
J_{\mathrm{CGPA}}(s_0)-J_F(s_0)
\leq{}& \sum_{i\in\mathcal A}\tau_i|V(s_i)|\\
 &+ (2+\tau_{\max})\epsilon |\mathcal A|
 + \sum_{j\in\mathcal R}\rho_j .
\end{aligned}
\end{equation}
\end{theorem}

The boundary error enters once per accepted segment and adds
$(2{+}\tau_{\max})\epsilon|\mathcal A|$ (full proof in the supplementary
material).

\begin{proposition}[Boundary error is charged per accepted segment]
\label{prop:eps-necessary}
The substitution proving Certificate~\ref{thm:approx} introduces boundary
error only at accepted-segment endpoints --- interior states telescope and
cancel --- so $\epsilon$ enters once per accepted segment, giving the
$O(\epsilon|\mathcal A|)$ term. The error cost thus scales with the number of
trusted-boundary crossings $|\mathcal A|$, not the interior steps a segment
spans: each charge stands in for the whole run of stepwise oracle solves the
segment covers, amortized over strictly more eliminated solves whenever a
segment spans multiple steps.
\end{proposition}

A single operator (Definition~\ref{def:operator}) forces these components
together once safety, bounded regret, and a learnable boundary are required;
remove any one and CGPA collapses onto a known weaker baseline
(Proposition~\ref{prop:decoupling}).
Section~\ref{sec:experiments} tests these structural consequences one by one.

\subsection{Certificate-Aware Learning of the Value Boundary}
\label{sec:learned}

We instantiate $\hat V$ with a learned five-quantile MLP that outputs levels
$\{\hat V_{0.1},\hat V_{0.3},\hat V_{0.5},\hat V_{0.7},\hat V_{0.9}\}$ ---
a calibrated central band with tail levels for the error spread. The median
drives the value guard \eqref{eq:value_guard} and the spread
$\hat V_{0.9}-\hat V_{0.1}$ supplies a state-dependent estimate of the
$\epsilon$ in Certificate~\ref{thm:approx}. The training objective is
certificate-driven: the only boundary-dependent term of
Certificate~\ref{thm:approx} is $(2+\tau_{\max})\epsilon|\mathcal A|$, and
calibration (Proposition~\ref{prop:conformal-bound}) makes $\epsilon(s)$ a
fixed multiple of the quantile spread $\hat V_{0.9}(s)-\hat V_{0.1}(s)$, so we
add a mean-spread penalty on top of the coverage loss:
\begin{equation}
\mathcal L \;=\; \underbrace{\mathcal L_{\text{pinball}}}_{\text{coverage}}
\;+\; \lambda \cdot
\underbrace{\mathbb E_s\bigl[\hat V_{0.9}(s)-\hat V_{0.1}(s)\bigr]}_{\text{mean }\epsilon\text{ estimate}} ,
\label{eq:cert_loss}
\end{equation}
the pinball term \cite{koenker1978regression} preserving quantile coverage and
the spread penalty shrinking Certificate~\ref{thm:approx}'s $\epsilon$ term.
The penalty is essential: removing it ($\lambda{=}0$) inflates EMS regret
by $1.46\times$ at matched call reduction (we use $\lambda=0.05$; full
sweep in the appendix). Because $\hat V$ enters only the value
guard, a poor boundary inflates regret but cannot induce a violation
(Propositions~\ref{prop:safety}--\ref{prop:decoupling}); learning is thus free
to optimize the frontier with no safety risk.

To instantiate the $\epsilon$ of Assumption~\ref{asm:bounded} as a
measured quantity, we apply split-conformal calibration
\cite{vovk2005algorithmic,lei2018distributionfree,angelopoulos2023gentle}
to the quantile-spread estimate on held-out oracle rollouts, with
non-conformity score
$r_i=|y_i-\hat V_{0.5}(s_i)|/(\hat V_{0.9}(s_i)-\hat V_{0.1}(s_i))$, where
$y_i=V(s_i)$ is the realized oracle cost-to-go at calibration state $s_i$. This
yields a state-dependent band
$[\hat V_{0.5}(s)\pm\alpha\,(\hat V_{0.9}(s)-\hat V_{0.1}(s))]$, calibrated
to $80\%$ and $90\%$ targets and attaining $0.816$--$0.853$ and $0.900$--$0.917$
empirical coverage across domains (appendix).
The calibrated multiplier $\alpha$ then defines a \emph{state-dependent}
error band $\epsilon(s)=\alpha\,(\hat V_{0.9}(s)-\hat V_{0.1}(s))$ that
instantiates the certificate's $\epsilon$ --- so the conformal band and the
certificate's error term are one quantity, and the spread penalty of
Eq.~\eqref{eq:cert_loss} shrinks precisely this band. Deployed as a gate, the
band enters the value guard pessimistically at both endpoints: a prefix is
admitted only if
\begin{equation}
\begin{aligned}
g(s_t,u_{0:k-1})+\hat V_{0.5}(s_{t+k})&+\epsilon(s_{t+k})\\
\leq \hat V_{0.5}(s_t)-\epsilon(s_t)&+\tau|\hat V_{0.5}(s_t)| ,
\end{aligned}
\label{eq:band_gate}
\end{equation}
the band-gate operating points of
Sections~\ref{sec:quality} and~\ref{sec:deploy}, yielding
a forward, high-probability guarantee.

\begin{proposition}[Conformal forward regret bound]
\label{prop:conformal-bound}
Let the calibration set be exchangeable with the deployed boundary states,
and choose the band multiplier $\alpha$ at miscoverage level $\delta/(T{+}1)$,
so that $\Pr[\,|\hat V_{0.5}(s)-V(s)|\le\epsilon(s)\,]\ge 1-\delta/(T{+}1)$ at
each reachable boundary state. Write
$\bar\epsilon_i=\max\{\epsilon(s_i),\epsilon(s_{i+1})\}$ for the endpoint
band of accepted segment $i$. Then, under Assumption~\ref{asm:fallback},
with probability at least $1-\delta$,
\begin{equation}
\begin{aligned}
J_{\mathrm{CGPA}}(s_0)-J_F(s_0)\le{}&
\sum_{i\in\mathcal A}\tau_i|V(s_i)|
+\sum_{j\in\mathcal R}\rho_j\\
&+(2+\tau_{\max})\sum_{i\in\mathcal A}\bar\epsilon_i .
\end{aligned}
\end{equation}
\end{proposition}

A union bound at level $\delta/(T{+}1)$ over the $\le T{+}1$ boundary states
makes $|\hat V_{0.5}(s)-V(s)|\le\epsilon(s)$ hold jointly with probability
$\ge 1-\delta$; the Certificate~\ref{thm:approx} substitution then applies
with the endpoint bands (full proof in the appendix).
The regret certificate is a realized-reference audit needing no
exchangeability assumption; the conformal band calibrates $\epsilon$ before
deployment.

\section{Experiments}
\label{sec:experiments}

\subsection{Setup}
\label{sec:setup}

We evaluate CGPA across three sequential-decision settings spanning a
spectrum of oracle cost: a fast LP-backed \emph{battery energy management}
(EMS) problem on UCI household loads \cite{ucipower2012} (the primary
mechanism testbed), the real-data \emph{CityLearn 2022} multi-building HVAC
benchmark \cite{vazquezcanteli2019citylearn} (a real-calendar distribution
shift, appendix), and a MILP-backed deployment-scale
\emph{unit commitment} (UC) instance \cite{brown2018pypsa} (an expensive
oracle). The contract is held fixed; the untrusted proposal slot is the
variable we sweep, to test the invariance of the certificate's guarantees. Across the study
the proposal slot is filled by a deterministic forecast, a learned MLP
policy, an end-to-end reinforcement learner, adversarial drafters, and six
frozen LLMs (Qwen3-$1.7$B/$4$B/$8$B; Gemma-4-E2B/E4B/12B), exercising the
arbitrariness of Assumption~\ref{asm:drafter}.

Unless stated otherwise EMS episodes use a $24$-step horizon and prefix
length $K{=}4$, and the value tolerance $\tau$ --- the per-segment
regret-budget parameter of Certificate~\ref{thm:approx} --- is held fixed
within a setting; its regret--call-reduction response is smooth and monotone
(appendix), so it directly selects the operating point.
Throughout, violations count applied constraint breaches, regret is
$R_{\mathrm{rel}}$ of Eq.~\eqref{eq:regret} (p95: its $95$th percentile), and
call reduction is the relative drop in trusted-oracle invocations versus
stepwise fallback. The learned boundary is the five-quantile
certificate-aware value boundary of Section~\ref{sec:learned}, run via
per-state online CPU inference inside the timed loop; confidence intervals
are nonparametric episode-level bootstrap. The experiments share one design,
\emph{controlled frontier slicing}: each table fixes all but one part of the
accept/defer system (value boundary, acceptance mechanism, proposal source,
or fallback oracle), so it slices a different term of
Certificate~\ref{thm:approx}.

\subsection{The Contract Neutralizes Arbitrary Proposal Sources}
\label{sec:safety}

Safety is the $|\mathcal A|\!\to\!0$ limit of the certificate: it never
queries $\epsilon$ (Proposition~\ref{prop:decoupling}), so the contract
drives \emph{any} untrusted source to zero applied violations --- a guarantee
no direct AI replacement provides. We fill the EMS proposal slot with
adversarial hand-built drafters, an end-to-end learner, and frozen LLMs,
comparing direct execution against the same source wrapped by CGPA
(appendix).

The pattern is uniform. Adversarial drafters and the two larger Gemma models
produce parseable but constraint-violating actions in $98$--$100\%$ of
episodes when executed directly; the contract admits none of those unsafe
transitions and defers, reaching zero applied violations while collapsing
regret by orders of magnitude --- the empirical counterpart of
Proposition~\ref{prop:safety}, where safety is owned by the verifier/oracle
pair regardless of how proposals are generated. The regret bound also holds:
a low-quality source has few prefixes accepted, so $|\mathcal A|$ --- and the
certificate's $(2+\tau_{\max})\epsilon|\mathcal A|$ term --- shrinks toward
zero and CGPA reverts to stepwise fallback at unchanged guarantees.

The end-to-end learner makes the converse point: replacing the oracle
with a trained policy removes oracle calls but pays $0.503$ mean regret ---
about $7\times$ the same learner used as a proposal source inside the contract
($0.074$). Direct AI replacement is not a substitute for the contract; it
forfeits either safety or regret.

\paragraph{Safety survives intentional misuse.} A white-box adaptive adversary
--- one that re-plans every step with full knowledge of the verifier, value
boundary, dynamics, and cost to maximize constraint violations --- drives raw
execution to a breach in every episode on both EMS and unit commitment, yet
CGPA admits none: it defers at every step, so the attack is absorbed by lost
amortization (call reduction falls to zero), never by an applied violation
(appendix). The feasibility gate never queries the value
boundary the adversary attacks, so safety is owned by the verifier, not the
proposal source, however adversarial. The same holds against a \emph{learned}
attacker: a white-box MLP that optimizes feasible high-cost prefixes against the
contract drives feasible-cost regret above the forecast drafter, yet the frozen
contract still holds it to zero applied violations with regret inside the
certificate (appendix).

\subsection{Learned Value Boundaries Lift the Risk--Cost Frontier}
\label{sec:quality}

Quality has one structural source: the certificate is \emph{one-sided}. It
prices the boundary error of accepted coverage at $(2+\tau_{\max})\epsilon$
per segment, bounds only the excess over the fallback, and never references
oracle optimality (Assumption~\ref{asm:fallback}). With safety owned by the
verifier, $\epsilon$ is a free optimization target: learning the boundary
down shrinks the price of coverage, and a capable source pushes accepted
quality past the stepwise reference.
We trace the EMS risk--cost frontier under a fixed forecast proposal source
($n{=}100$), changing only the value boundary (appendix). Plain
prefix acceptance --- a multi-step safe filter with no value guard --- accepts
everything and pays $0.4450$ mean regret, identically at every draft
length $K$ (appendix); end-to-end RL replacement pays
$0.503$ (appendix). The certificate-aware learned guard
cuts this to $0.0063$ at $60.4\%$ call reduction, and its conformally
calibrated band gate drives mean regret to $\mathbf{0.0002}$ --- three
orders of magnitude below unguarded acceptance, within sampling noise of the
stepwise oracle ($95\%$ CI $[-0.005,0.006]$, spanning zero) --- the
lowest mean regret and best p95 across the EMS value-boundary sweep at $\tau{=}0.04$. It also outperforms non-learned
boundaries and reproduces on the real-data CityLearn benchmark ($0.0104$ mean
regret, zero violations; appendix). This is the certificate's
$(2+\tau_{\max})\epsilon|\mathcal A|$ term made small, and the
realized-reference audit confirms it: the envelope covers $93$--$100\%$ of
accepted segments, at zero applied violations throughout.

\paragraph{Certified acceptance versus full-plan approximation.} The next
slice fixes boundary and source and varies the acceptance mechanism. We
route an imitation network trained on oracle sequences --- the
safety-augmented neural MPC baseline of \citet{hose2023approxnmpc} --- through
two controllers. Full-plan
acceptance commits the whole sequence as one accept/reject decision and
repairs to a safe candidate on infeasibility, paying $0.0227$ mean regret;
CGPA's certified prefix acceptance, on the same network, truncates at the
value boundary and cuts regret to $0.0081$ ($2.8\times$) at $86.6\%$ call
reduction, or holds a $98.0\%$ high-amortization point at a fraction of the
verifier load (appendix). This gap is the per-segment
$(2+\tau_{\max})\epsilon$ pricing of Certificate~\ref{thm:approx} at fixed
boundary and source.

\paragraph{Distribution shift: the stepwise oracle is overtaken.} The learned
boundary generalizes beyond its training distribution. We swap in a learned
MLP \emph{drafter} as an out-of-distribution (OOD) proposal source under
calendar shift: it is trained on one season and deployed on $18$ held-out
OOD days, against the forecast drafter on the same days (five training
seeds), with the certified contract unchanged.
Table~\ref{tab:ood} reports the result. The learned proposal source raises
mean call reduction by $+18.8$pp (day-clustered $95\%$ CI $[+14.9,+22.1]$,
bootstrapped over held-out days) while \emph{lowering} mean regret by $0.055$,
at zero violations across all $90$ day--seed runs.
The sign is decisive: the accepted prefixes beat the stepwise oracle's
own cost on $15$ of $18$ days ($72$ of $90$ day--seed runs; supplementary
material), while in distribution they win exactly $50\%$ --- the expected rate
against a matched reference. Under shift the contract overtakes it by
committing multi-step prefixes its myopic per-step replanning cannot.

\begin{table}[t]
\centering
\caption{Out-of-distribution learned proposal generalization on EMS
($18$ held-out calendar-shifted days, five seeds, $90$ paired day--seed
runs); deltas are paired with $95\%$ bootstrap CIs. Negative regret means the
accepted prefixes outperform the oracle's stepwise plan; ``days $<$ oracle''
counts held-out days on which the contract beats the stepwise oracle's own
cost. Zero violations throughout; $\tau{=}0.12$.}
\label{tab:ood}
\scriptsize
\setlength{\tabcolsep}{2pt}
\begin{tabular*}{\columnwidth}{@{\extracolsep{\fill}}lrrrr@{}}
\toprule
Proposal source & Call red. & Mean reg. & Days $<$ oracle & Viol. \\
\midrule
Forecast (OOD)          & $57.4\%$ & $-0.082$ & $12/18$ & $0$ \\
Learned (OOD)           & $76.2\%$ & $-0.137$ & $15/18$ & $0$ \\
\midrule
$\Delta$ (L$-$F) & $+18.8$pp & $-0.055$ & --- & $0$ \\
\quad 95\% CI  & $[+14.9,+22.1]$ & $[-0.098,-0.012]$ & --- & --- \\
\bottomrule
\end{tabular*}
\end{table}

\subsection{One Contract, Six Heterogeneous LLM Proposal Sources}
\label{sec:diversity}

Frozen LLMs are the most heterogeneous proposal sources available --- varying
in scale, family, and decoding behavior with no domain-specific training.
Because the certificate charges per accepted segment, such a source moves only
\emph{which} prefixes are accepted, not the contract's guarantees.
Table~\ref{tab:llm} fills the EMS proposal slot with six of them (same-seed
$n{=}50$, deterministic decoding, $\tau{=}0.08$).

\begin{table*}[t]
\centering
\caption{Six heterogeneous frozen LLM proposal sources on battery EMS
(same-seed $n{=}50$, $\tau{=}0.08$, deterministic decoding). The same
certified contract absorbs every model at \emph{zero} applied violations with
one shared prompt template. Bracketed values are
$95\%$ bootstrap CIs over the $50$ episodes; the ratio is omitted where
direct execution is violation-dominated.}
\label{tab:llm}
\scriptsize
\setlength{\tabcolsep}{5pt}
\begin{tabular*}{\textwidth}{@{\extracolsep{\fill}}lrrrrrr@{}}
\toprule
Model & Direct reg. [CI] & CGPA reg. [CI] & Ratio & CGPA call red. [CI] & Direct viol. & CGPA viol. \\
\midrule
Qwen3-1.7B   & $0.421\,[0.408,0.435]$ & $0.068\,[0.053,0.083]$ & $6.23\times$ & $68.1\%\,[66.3,69.8]$ & $0$ & $0$ \\
Qwen3-4B     & $0.353\,[0.341,0.366]$ & $0.061\,[0.048,0.074]$ & $5.80\times$ & $69.8\%\,[67.8,71.8]$ & $0$ & $0$ \\
Qwen3-8B     & $0.349\,[0.337,0.362]$ & $0.032\,[0.021,0.045]$ & $10.94\times$ & $65.5\%\,[63.5,67.5]$ & $0$ & $0$ \\
Gemma-4-E2B  & $0.362\,[0.344,0.383]$ & $0.094\,[0.082,0.107]$ & $3.84\times$ & $61.5\%\,[58.8,64.3]$ & $0$ & $0$ \\
Gemma-4-E4B  & $200.0\,[195.9,204.4]$ & $0.020\,[0.014,0.027]$ & --- & $53.4\%\,[51.2,55.7]$ & $1.00$ & $0$ \\
Gemma-4-12B  & $20.38\,[19.09,21.46]$ & $0.032\,[0.022,0.041]$ & --- & $66.7\%\,[65.2,68.3]$ & $0.98$ & $0$ \\
\bottomrule
\end{tabular*}
\end{table*}

Two regimes appear, and the contract absorbs both. The Qwen models and
Gemma-4-E2B are feasible but costly under direct execution; the contract
sharpens them, cutting regret $3.8$--$10.9\times$ at $62$--$70\%$ call
reduction. Gemma-4-E4B and Gemma-4-12B violate constraints in nearly every
direct rollout; the contract rejects those transitions and defers to zero
violations at low regret. One contract thus delivers the same guarantee across
a $1.7$B and a $12$B model alike: the proposal source determines which
prefixes are accepted --- certified throughput tracks proposal quality,
not parameter count --- never whether the system is safe. A controlled quality
sweep makes this explicit: as a controlled knob moves the source from random
toward oracle, EMS call reduction rises monotonically ($24\%\!\to\!90\%$,
Spearman $1.0$) at zero violations throughout, and the same quality--throughput
ordering holds on UC at matched tolerance (appendix). The same contract that absorbs all six
frozen LLMs at zero violations is next evaluated on an oracle expensive enough
that each skipped call yields measurable wall-clock savings.

\subsection{Deployment-Scale Wall-Clock Amortization}
\label{sec:deploy}

\begin{figure}[t]
\centering
\includegraphics[width=0.80\columnwidth]{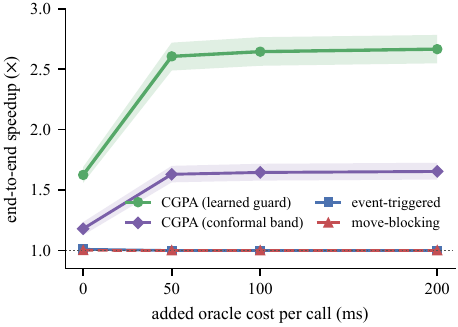}
\caption{Amortization is structural: counterfactually scaling the per-call
oracle cost on EMS lifts CGPA's end-to-end speedup from $1.63\times$ to
$2.67\times$ (conformal band gate: $1.18\times$ to $1.65\times$),
while event-triggered and move-blocking controllers are
pinned near $1.0\times$ once the oracle carries any real cost; zero
violations throughout. Shaded bands: $95\%$ bootstrap CIs ($n{=}100$).}
\label{fig:amortization}
\end{figure}

Wall-clock savings scale with skipped oracle solves: scaling per-call oracle
cost on EMS lifts CGPA's speedup from $1.63\times$ to $2.67\times$, while
event-triggered and move-blocking controllers stay near $1.0\times$
(Figure~\ref{fig:amortization}). On a deployment-scale $48$-step UC MILP,
Table~\ref{tab:uc} reports a hold-commitment drafter and a frozen Qwen3-8B
proposing commitment changes against an event-triggered receding-horizon
controller that re-solves on deviation ($n{=}50$), both CGPA rows using the
\emph{same} learned boundary of Section~\ref{sec:learned}.

\begin{table}[t]
\centering
\caption{Deployment-scale unit commitment ($48$-step horizon, $n{=}50$). MPC
rows report each window's strongest trigger configuration; call reduction
applies to accept/defer controllers; speedup CIs are episode-level bootstrap.
The hold drafter uses $\tau{=}0.04$, the Qwen drafter $\tau{=}0.15$.}
\label{tab:uc}
\scriptsize
\setlength{\tabcolsep}{3.5pt}
\begin{tabular*}{\columnwidth}{@{\extracolsep{\fill}}lrrrr@{}}
\toprule
Controller & Speedup [CI] & Call red. & Reg. & Viol. \\
\midrule
CGPA --- Qwen3-8B & $2.96\times\,[2.01,4.31]$ & $64.3\%$ & $0.021$ & $0$ \\
CGPA --- hold & $1.79\times\,[1.77,1.82]$ & $63.6\%$ & $0.015$ & $0$ \\
Event-trig.\ MPC ($24$) & $1.07\times\,[0.997,1.15]$ & --- & $0.012$ & $0$ \\
Event-trig.\ MPC ($48$) & $0.84\times\,[0.79,0.88]$ & --- & $0.012$ & $0$ \\
\bottomrule
\end{tabular*}
\end{table}

The frozen LLM, with no UC-specific training, achieves the highest speedup: its
commitment changes admit prefixes the deterministic heuristic does not propose,
and the contract converts them into a $2.96\times$ per-episode mean speedup
($2.02\times$ in aggregate wall-clock, every one of the $50$ episodes faster
than fallback) at $2.1\%$ regret --- aggregate amortization the hold drafter
($1.79\times$) does not reach --- while the contract, not the proposal source,
owns the unchanged zero-violation guarantee. Trigger-scheduled replanning achieves at most $1.07\times$; warm-starting
reaches $0.91\times$. The MILP fallback dominates online time
($97.96$--$99.98\%$); a matched-$\tau$ rerun and a second instance
($2.09\times$) confirm the ordering (appendix).

\section{Related Work}

\paragraph{AI control and scalable oversight.} A parallel safety agenda wraps a
capable but \emph{untrusted} policy in a trusted protocol and measures safety
under adversarial control evaluation \cite{greenblatt2024aicontrol}, as do
weak-to-strong generalization \cite{burns2023weaktostrong} and debate
\cite{irving2018debate}, which separate a strong untrusted generator from a
weaker trusted overseer under the asymmetry that verification is cheaper than
generation. CGPA is the certified, sequential-control instance of this
separation --- intelligence and guarantees come from different owners --- but
its trusted core gives a \emph{deterministic} feasibility guarantee and an
auditable per-segment regret price rather than a monitored probability of
subversion, accepting the source by exact verification and a calibrated value
certificate rather than a learned monitor.

\paragraph{Speculative decoding.} CGPA shares the draft--verify--accept
skeleton of speculative decoding
\cite{leviathan2023speculative,chen2023speculative,cai2024medusa,li2024eagle}
but changes the verified object from a token distribution to hard transition
safety plus a certified fallback-relative cost, and the proposal source from
a draft model to an arbitrary AI policy or LLM. Recent production variants add
a calibrated confidence head and a load-adaptive verification schedule
\cite{deepseekai2026dspark}, yet acceptance stays \emph{lossless} ---
calibrated to preserve the target token distribution --- whereas CGPA
calibrates an analogous value boundary to certify decision cost under a
hard-feasibility verifier. Speculative execution has
also cut agentic latency by predicting actions and committing on match
\cite{ye2025speculativeactions}; CGPA instead accepts \emph{deviations} from
the oracle within a certified regret budget, under hard feasibility rather
than equivalence checking.

\paragraph{Selective prediction and learning to defer.} The accept/defer
boundary is a learning-to-defer mechanism for the control loop --- acceptance
is coverage, deferral is abstention to a trusted expert, and the certificate
bounds the cost of coverage
\cite{geifman2017selective,mozannar2020consistent,madras2018predict}. Unlike
sequence-level deferral for token and span abstention
\cite{rayan2025partialdefer}, CGPA defers \emph{executed control prefixes}
under a hard-feasibility verifier and a fallback-relative regret certificate.

\paragraph{Conformal risk control.} The calibrated band places CGPA within
conformal risk control
\cite{vovk2005algorithmic,angelopoulos2023gentle,angelopoulos2024conformalrisk,prinster2026conformalpolicy}:
the controlled risk is the per-segment excess cost the certificate charges,
making $\epsilon$ measured rather than assumed.

\paragraph{Safe control and safe AI execution.} Shielding and predictive safety
filters
\cite{alshiekh2018shielding,wabersich2021predictive,leeman2023predictive,ames2017control}
correct a single action per online solve; intervention controllers pair a
learned proposer with a trusted backup \cite{wagener2021sailr,fulton2018safe};
safety-augmented neural MPC falls back on infeasibility
\cite{hose2023approxnmpc}; other lines accelerate or repair the trusted
controller rather than skip it
\cite{yu2025neural,tabuada2007event,gondhalekar2010moveblocking,diehl2002realtime,amos2018differentiable,bertsekas1997rollout}.
CGPA admits an arbitrary source, separates exact feasibility from a learned
value boundary, and attaches a calibrated telescoping regret certificate.

\section{Conclusion}

CGPA is a certified speculative-execution contract: a thin trusted core
owns every guarantee while an arbitrarily capable untrusted source is
admitted only through it. The same contract drives adversarial drafters and
six frozen LLMs to zero applied violations, holds regret within sampling
noise of the stepwise oracle, and on deployment-scale unit commitment turns
a frozen LLM into a $2.96\times$ wall-clock speedup --- turning
heterogeneous untrusted AI agents into a safe, low-regret,
oracle-amortizing execution layer whose throughput grows with the capability
of the untrusted source, at guarantees that never change.

\bibliography{references}

\clearpage
\appendix
\section{Notation}
Table~\ref{tab:notation} collects the symbols used throughout the main paper
and this appendix.

\begin{table*}[t]
\centering
\caption{Notation.}
\label{tab:notation}
\small
\begin{tabular}{@{}l L{0.82\textwidth}@{}}
\toprule
Symbol & Meaning \\
\midrule
$s_t,\,a_t$ & state and applied action at step $t$ \\
$f,\,c$ & transition map $s_{t+1}{=}f(s_t,a_t)$ and stage cost \\
$\pi_F$ & trusted fallback (oracle) policy \\
$V{=}V^{\pi_F}$ & finite-horizon fallback cost-to-go (terminal value $0$) \\
$D(s,K)$ & untrusted source proposing a length-$K$ action prefix \\
$\hat V$ & cheap value boundary (precomputed or learned), queried in the value guard \\
$\tau,\,\tau_i$ & value tolerance (global; on accepted segment $i$) \\
$\epsilon$ & boundary-error bound $|\hat V(s)-V(s)|\le\epsilon$ at accepted endpoints \\
$\bar\epsilon_i$ & per-segment boundary error: $\max$ of $\epsilon$ over the segment endpoints \\
$g_i$ & realized cost of accepted segment $i$ \\
$f_j,\,\rho_j$ & cost and realized slack of fallback repair segment $j$ \\
$\mathcal A,\,\mathcal R$ & sets of accepted and fallback-repair segments \\
$J_F,\,J_{\mathrm{CGPA}}$ & stepwise-fallback and CGPA episode cost \\
$R_{\mathrm{rel}}$ & fallback-relative normalized regret \\
\bottomrule
\end{tabular}
\end{table*}

\section{Certificate-Aware Loss Sweep}
This section reports the full $\lambda$ sweep referenced in Section 2.5. The spread penalty is essential: at matched call reduction, $\lambda=0.05$ reduces EMS mean regret from $0.0531$ to $0.0364$.

\begin{table}[t]
\centering
\caption{$\lambda$ sweep for the certificate-aware boundary (EMS $K{=}4$, $\tau{=}0.12$; CityLearn $K{=}6$, $\tau{=}0.04$). Regret and call reduction are reported as mean $\pm$ standard deviation over ten seeds.}
\label{tab:supp_lambda}
\scriptsize
\setlength{\tabcolsep}{4pt}
\begin{tabular*}{\columnwidth}{@{\extracolsep{\fill}}llrrrr@{}}
\toprule
Domain & $\lambda$ & Seeds & Regret & Call red. & Viol. \\
\midrule
EMS & $0$ & $10$ & $0.0531\pm0.0246$ & $69.3\pm1.3\%$ & $0$ \\
CityLearn & $0$ & $10$ & $0.0250\pm0.0024$ & $88.0\pm4.0\%$ & $0$ \\
EMS & $0.05$ & $10$ & $0.0364\pm0.0116$ & $68.9\pm0.8\%$ & $0$ \\
CityLearn & $0.05$ & $10$ & $0.0246\pm0.0030$ & $85.0\pm3.3\%$ & $0$ \\
EMS & $0.30$ & $10$ & $-0.0023\pm0.0033$ & $40.1\pm5.2\%$ & $0$ \\
CityLearn & $0.30$ & $10$ & $0.0161\pm0.0030$ & $13.0\pm5.0\%$ & $0$ \\
\bottomrule
\end{tabular*}
\end{table}

\section{Conformal Calibration}
This section reports the calibration quantities supporting the coverage statement in Section 2.5. Calibration and test states are seed-disjoint fallback-reachable rollout-state visits ($n_{\mathrm{cal}}{=}n_{\mathrm{test}}{=}408$ per domain). The nonconformity score is
\[
r_i=\frac{|y_i-\hat V_{0.5}(s_i)|}{\hat V_{0.9}(s_i)-\hat V_{0.1}(s_i)},
\]
and the deployed band is $[\hat V_{0.5}(s)\pm\alpha(\hat V_{0.9}(s)-\hat V_{0.1}(s))]$. At this calibration size both domains meet their $80\%$ and $90\%$ targets.

\begin{table*}[t]
\centering
\caption{Split-conformal calibration of the value-boundary band on seed-disjoint rollout states ($n_{\mathrm{cal}}{=}n_{\mathrm{test}}{=}408$). The two target miscoverage levels instantiate the 80\% and 90\% bands.}
\label{tab:supp_conformal}
\scriptsize
\setlength{\tabcolsep}{3pt}
\begin{tabular*}{\textwidth}{@{\extracolsep{\fill}}lrrrrrrrr@{}}
\toprule
Domain & Calib. & Test & $\delta_{80}$ & $\alpha_{80}$ & Cov.80 & $\delta_{90}$ & $\alpha_{90}$ & Cov.90 \\
\midrule
EMS & $408$ & $408$ & $0.20$ & $4.104$ & $0.816$ & $0.10$ & $4.318$ & $0.900$ \\
CityLearn & $408$ & $408$ & $0.20$ & $1.444$ & $0.853$ & $0.10$ & $1.777$ & $0.917$ \\
\bottomrule
\end{tabular*}
\end{table*}

\section{Tolerance Frontier}
This section reports the $\tau$ frontier referenced in Section 3.1. Increasing $\tau$ selects longer accepted prefixes and higher call reduction with zero applied violations across the reported sweep.

\begin{table}[t]
\centering
\caption{$\tau$ frontier for EMS and CityLearn. Call reduction and regret are medians over ten seeds; mean prefix length reports the accepted-prefix operating point.}
\label{tab:supp_tau}
\scriptsize
\setlength{\tabcolsep}{4pt}
\begin{tabular*}{\columnwidth}{@{\extracolsep{\fill}}lrrrrrr@{}}
\toprule
Domain & $\tau$ & Seeds & Call red. & Regret & Prefix & Viol. \\
\midrule
EMS & $0.01$ & $10$ & $46.1\%$ & $-0.0044$ & $0.83$ & $0$ \\
EMS & $0.02$ & $10$ & $52.6\%$ & $-0.0027$ & $1.00$ & $0$ \\
EMS & $0.04$ & $10$ & $60.0\%$ & $0.0060$ & $1.20$ & $0$ \\
EMS & $0.08$ & $10$ & $66.1\%$ & $0.0129$ & $1.38$ & $0$ \\
EMS & $0.12$ & $10$ & $68.8\%$ & $0.0315$ & $1.47$ & $0$ \\
EMS & $0.20$ & $10$ & $83.0\%$ & $0.2065$ & $2.33$ & $0$ \\
CityLearn & $0.01$ & $10$ & $81.4\%$ & $0.0239$ & $3.29$ & $0$ \\
CityLearn & $0.02$ & $10$ & $83.5\%$ & $0.0240$ & $3.47$ & $0$ \\
CityLearn & $0.04$ & $10$ & $85.6\%$ & $0.0242$ & $3.72$ & $0$ \\
CityLearn & $0.08$ & $10$ & $88.0\%$ & $0.0246$ & $3.99$ & $0$ \\
CityLearn & $0.12$ & $10$ & $89.6\%$ & $0.0252$ & $4.16$ & $0$ \\
CityLearn & $0.20$ & $10$ & $90.9\%$ & $0.0265$ & $4.33$ & $0$ \\
\bottomrule
\end{tabular*}
\end{table}

\section{Certified Execution over Untrusted Proposal Sources}
This section gives the per-source safety/regret table referenced in Section 3.2. \emph{Direct} executes the proposed actions raw; \emph{CGPA} (Certificate-Gated Prefix Acceptance) wraps the same source in the accept/defer contract (the RL \emph{Direct} column is end-to-end oracle replacement). Adversarial drafters and the frozen Gemma-4-12B and Gemma-4-E4B models violate constraints in nearly every episode under direct execution; the contract drives all of them to zero applied violations and collapses regret by three to four orders of magnitude, the worst adversarial source landing at the oracle's own cost. Violations are the per-episode applied-breach rate; adversarial/RL rows use $30$ EMS episodes, LLM rows $50$ same-seed episodes ($\tau{=}0.08$); the full six-LLM comparison appears in the main paper.

\begin{table}[t]
\centering
\caption{Certified execution over untrusted proposal sources on battery EMS. \emph{Direct} executes the proposed actions raw; \emph{CGPA} wraps the same source in the accept/defer contract.}
\label{tab:supp_safety}
\scriptsize
\setlength{\tabcolsep}{4.5pt}
\begin{tabular*}{\columnwidth}{@{\extracolsep{\fill}}lrrrr@{}}
\toprule
 & \multicolumn{2}{c}{Direct} & \multicolumn{2}{c}{CGPA} \\
\cmidrule(lr){2-3}\cmidrule(lr){4-5}
Proposal source & Viol. & Regret & Viol. & Regret \\
\midrule
Adversarial (always-infeasible)        & $1.00$ & $246.9$ & $0$ & $-0.0001$ \\
Adversarial (always-charge)& $1.00$ & $247.9$ & $0$ & $0.119$ \\
Adversarial (anti-forecast)   & $1.00$ & $247.1$ & $0$ & $0.042$ \\
Adversarial (random)     & $1.00$ & $189.7$ & $0$ & $0.114$ \\
End-to-end RL                     & $0$    & $0.503$ & $0$ & $0.074$ \\
Frozen LLM Gemma-4-12B           & $0.98$ & $20.38$ & $0$ & $0.032$ \\
Frozen LLM Gemma-4-E4B           & $1.00$ & $200.0$ & $0$ & $0.020$ \\
\bottomrule
\end{tabular*}
\end{table}

\section{EMS Value-Boundary Frontier}
This section gives the EMS risk--cost frontier referenced in Section 3.3, under a fixed forecast proposal source ($n{=}100$) with only the value boundary changed ($\tau{=}0.04$). The certificate-aware learned guard cuts mean regret from $0.4450$ (unguarded prefix acceptance) to $0.0063$ at $60.4\%$ call reduction, and its conformally calibrated band gate reaches $0.0002$ ($95\%$ CI $[-0.005,0.006]$, spanning zero) --- three orders of magnitude below unguarded acceptance, within sampling noise of the stepwise oracle --- at the best regret tail, all at zero applied violations.

\begin{table}[t]
\centering
\caption{EMS selective accept/defer risk--cost frontier under a fixed forecast proposal source ($n{=}100$); only the value boundary changes. Call reduction is the relative drop in trusted-oracle invocations versus stepwise fallback; regret is $R_{\mathrm{rel}}$; $\tau{=}0.04$.}
\label{tab:supp_mechanism}
\scriptsize
\setlength{\tabcolsep}{5pt}
\begin{tabular*}{\columnwidth}{@{\extracolsep{\fill}}lrrrr@{}}
\toprule
Value boundary & Call red. & Reg. & p95 reg. & Viol. \\
\midrule
Unguarded prefix              & $100.0\%$ & $0.4450$ & $0.5817$ & $0$ \\
Learned guard (cert-aware)    & $60.4\%$  & $0.0063$ & $0.0953$ & $0$ \\
\;+ conformal band gate       & $37.2\%$  & $\mathbf{0.0002}$ & $0.0465$ & $0$ \\
\bottomrule
\end{tabular*}
\end{table}

\section{Full-Plan Neural MPC Comparison}
This section reports the certified-acceptance-versus-full-plan comparison referenced in Section 3.3. The same imitation network (a safety-augmented neural MPC input-sequence approximator) is routed through full-plan acceptance --- which commits the whole predicted sequence as one accept/reject decision and repairs to a safe candidate on infeasibility --- and through CGPA certified prefix acceptance, with only the acceptance mechanism changed (EMS, $n{=}100$, zero violations throughout). CGPA prefix acceptance cuts mean regret from $0.0227$ to $0.0081$ ($2.8\times$) at $86.6\%$ call reduction, and trades regret for amortization at higher $\tau$.

\begin{table}[t]
\centering
\caption{Same neural proposal source (an imitation network), EMS $n{=}100$: certified prefix acceptance versus full-plan approximation, with only the acceptance mechanism changed. Call reduction is relative to stepwise fallback; regret is $R_{\mathrm{rel}}$; zero violations throughout.}
\label{tab:supp_neuralmpc}
\scriptsize
\setlength{\tabcolsep}{5pt}
\begin{tabular*}{\columnwidth}{@{\extracolsep{\fill}}lrrr@{}}
\toprule
Acceptance mechanism & Call red. & Regret & Verifier calls \\
\midrule
Full-plan accept (neural MPC) & $100.0\%$ & $0.0227$ & $48.0$ \\
CGPA prefix ($\tau{=}0.04$)   & $86.6\%$  & $\mathbf{0.0081}$ & $11.4$ \\
CGPA prefix ($\tau{=}0.40$)   & $98.0\%$  & $0.0304$ & $6.9$ \\
\bottomrule
\end{tabular*}
\end{table}

\section{CityLearn Real-Calendar Shift}
This section gives the real-calendar CityLearn result referenced in Sections 3.1 and 3.3. The same accept/defer contract is evaluated under the calendar shift with $\tau=0.04$.

\begin{table*}[t]
\centering
\caption{CityLearn real-calendar shift. The learned conformal band delivers the lowest mean regret at zero violations under the fixed contract.}
\label{tab:supp_citylearn}
\scriptsize
\setlength{\tabcolsep}{4pt}
\begin{tabular*}{\textwidth}{@{\extracolsep{\fill}}lrrrrrrr@{}}
\toprule
Boundary & $\tau$ & $n$ & Prefix & Call red. & Mean reg. & p95 reg. & Viol. \\
\midrule
Learned boundary & $0.04$ & $30$ & $3.39$ & $81.0\%$ & $0.0156$ & $0.0848$ & $0$ \\
Learned conformal band & $0.04$ & $30$ & $1.23$ & $57.5\%$ & $0.0104$ & $0.0848$ & $0$ \\
\bottomrule
\end{tabular*}
\end{table*}

\section{Oracle-Cost Amortization Sweep}
This section reports the full oracle-cost sweep supporting the deployment-amortization statement in Section 3.5. Increasing the per-call oracle cost demonstrates the structural advantage of certified multi-step acceptance over controllers that still call the oracle every step. Timing follows the per-component online accounting (draft, verification, boundary inference, trusted-oracle calls). The event-triggered controller replans when its queued plan is exhausted or when the next queued action fails one-step verification; because that plan is itself built from per-step trusted-oracle calls, its oracle-call count is unreduced under the target-oracle accounting used throughout.

\begin{table*}[t]
\centering
\caption{Oracle-cost amortization sweep on EMS ($n{=}100$). CGPA reduces oracle calls and its speedup rises with oracle cost; event-triggered and move-blocking controllers are pinned at $1.00\times$ once oracle cost dominates.}
\label{tab:supp_oracle_cost}
\scriptsize
\setlength{\tabcolsep}{4pt}
\begin{tabular*}{\textwidth}{@{\extracolsep{\fill}}lrrrrrrr@{}}
\toprule
Controller & Delay ms & $n$ & Speedup & CI & Oracle call red. & Mean reg. & Viol. \\
\midrule
CGPA learned guard & $0$ & $100$ & $1.63\times$ & $[1.56,1.69]\times$ & $60.4\%$ & $0.0063$ & $0$ \\
CGPA conformal band & $0$ & $100$ & $1.18\times$ & $[1.13,1.24]\times$ & $37.2\%$ & $0.0002$ & $0$ \\
Event-triggered & $0$ & $100$ & $1.01\times$ & $[0.99,1.03]\times$ & $0.0\%$ & $0.0000$ & $0$ \\
Move-blocking & $0$ & $100$ & $1.00\times$ & $[0.99,1.02]\times$ & $0.0\%$ & $0.0000$ & $0$ \\
CGPA learned guard & $50$ & $100$ & $2.61\times$ & $[2.49,2.72]\times$ & $60.4\%$ & $0.0063$ & $0$ \\
CGPA conformal band & $50$ & $100$ & $1.63\times$ & $[1.56,1.70]\times$ & $37.2\%$ & $0.0002$ & $0$ \\
Event-triggered & $50$ & $100$ & $1.00\times$ & $[1.00,1.00]\times$ & $0.0\%$ & $0.0000$ & $0$ \\
Move-blocking & $50$ & $100$ & $1.00\times$ & $[1.00,1.00]\times$ & $0.0\%$ & $0.0000$ & $0$ \\
CGPA learned guard & $100$ & $100$ & $2.65\times$ & $[2.53,2.77]\times$ & $60.4\%$ & $0.0063$ & $0$ \\
CGPA conformal band & $100$ & $100$ & $1.65\times$ & $[1.58,1.72]\times$ & $37.2\%$ & $0.0002$ & $0$ \\
Event-triggered & $100$ & $100$ & $1.00\times$ & $[1.00,1.00]\times$ & $0.0\%$ & $0.0000$ & $0$ \\
Move-blocking & $100$ & $100$ & $1.00\times$ & $[1.00,1.00]\times$ & $0.0\%$ & $0.0000$ & $0$ \\
CGPA learned guard & $200$ & $100$ & $2.67\times$ & $[2.55,2.78]\times$ & $60.4\%$ & $0.0063$ & $0$ \\
CGPA conformal band & $200$ & $100$ & $1.65\times$ & $[1.59,1.73]\times$ & $37.2\%$ & $0.0002$ & $0$ \\
Event-triggered & $200$ & $100$ & $1.00\times$ & $[1.00,1.00]\times$ & $0.0\%$ & $0.0000$ & $0$ \\
Move-blocking & $200$ & $100$ & $1.00\times$ & $[1.00,1.00]\times$ & $0.0\%$ & $0.0000$ & $0$ \\
\bottomrule
\end{tabular*}
\end{table*}

\section{UC Per-Episode Speedup}
All $50$ Qwen3-8B CGPA episodes behind the UC speedup statement in
Section 3.5 is faster than stepwise fallback, with a minimum per-episode speedup
of $1.51\times$, a per-episode mean of $2.96\times$, and a total-wall ratio of
$2.02\times$, all at zero applied violations.

\section{Proposal Quality and Matched-$\tau$ Throughput}
This section supports the claim that certified throughput grows with proposal \emph{quality}, not raw model size. On the six-LLM EMS suite (same-seed $n{=}50$), the verifier-accepted prefix length and the CGPA call reduction are perfectly rank-correlated across models (Spearman $\rho{=}1.000$), whereas parameter count and call reduction are not ($\rho{=}-0.058$): the capability axis that predicts throughput is the quality of the proposed prefixes, not the number of parameters.

The same ordering holds at deployment scale once the value tolerance is controlled. Table~\ref{tab:supp_matched_tau} re-runs the UC end-to-end comparison with the hold and Qwen3-8B drafters at \emph{identical} $\tau$ (the main-paper rows use $\tau{=}0.04$ for hold and $\tau{=}0.15$ for Qwen). At every matched $\tau$ the frozen LLM exceeds the deterministic heuristic on mean and total-wall speedup at zero violations: the capability--throughput ordering holds at matched tolerance.

\begin{table}[t]
\centering
\caption{UC end-to-end speedup with the hold and Qwen3-8B drafters at matched $\tau$ ($n{=}50$ per row, zero violations throughout). The Hold and Qwen3-8B columns are total-wall (aggregate) speedups; the last column is the per-episode mean difference (episode-bootstrap $95\%$ CI). Both favor Qwen3-8B at every matched tolerance.}
\label{tab:supp_matched_tau}
\scriptsize
\setlength{\tabcolsep}{5pt}
\begin{tabular*}{\columnwidth}{@{\extracolsep{\fill}}lrrr@{}}
\toprule
$\tau$ & Hold & Qwen3-8B & Qwen$-$hold mean $\Delta$ [95\% CI] \\
\midrule
$0.04$ & $1.79\times$ & $1.85\times$ & $+0.59\,[0.002,1.63]$ \\
$0.08$ & $1.78\times$ & $2.03\times$ & $+1.60\,[0.45,3.12]$ \\
$0.15$ & $1.79\times$ & $2.02\times$ & $+1.16\,[0.21,2.53]$ \\
\bottomrule
\end{tabular*}
\end{table}

A controlled quality knob isolates the same axis without changing the model: at
each step a synthetic EMS drafter takes the oracle/fallback action with
probability $q$ and a random feasible action otherwise, sweeping $q$ from random
($0$) to near-oracle ($1$) with the trusted core (boundary, verifier, $\tau$,
$K$, fallback) frozen at its deployment setting; $q{=}1$ reads a precomputed
oracle trajectory and serves as the controlled upper-bound endpoint. As
proposal quality rises, call reduction increases monotonically
($24\%\!\to\!90\%$, Spearman $\rho{=}1.0$) and mean regret decreases
monotonically ($0.031\!\to\!0.002$), at zero applied violations throughout
(Table~\ref{tab:supp_quality_sweep}): certified throughput tracks proposal
quality directly, consistent with parameter count being non-predictive across
the six-LLM suite ($\rho{=}{-}0.058$).

\begin{table}[t]
\centering
\caption{Controlled EMS proposal-quality sweep ($q$ from random to near-oracle;
$q{=}1$ is a controlled oracle upper bound; $n{=}100$, zero applied violations
throughout). Call reduction rises and mean regret falls monotonically with
proposal quality.}
\label{tab:supp_quality_sweep}
\scriptsize
\setlength{\tabcolsep}{6pt}
\begin{tabular*}{\columnwidth}{@{\extracolsep{\fill}}lrr@{}}
\toprule
$q$ & Call\,red. & Regret \\
\midrule
$0.0$ & $24\%$ & $0.031$ \\
$0.2$ & $38\%$ & $0.025$ \\
$0.4$ & $50\%$ & $0.024$ \\
$0.6$ & $63\%$ & $0.016$ \\
$0.8$ & $75\%$ & $0.007$ \\
$1.0$ & $90\%$ & $0.002$ \\
\bottomrule
\end{tabular*}
\end{table}

\section{UC Second-Instance Replication}
This section reports the second unit-commitment instance referenced in Section 3.5: $100$ generators, demand-peak fraction $0.66$, $48$-step horizon, hold-commitment drafter, and a certificate-aware value boundary trained on this instance ($40$ training episodes; $50$ evaluation episodes on held-out seeds, $\tau{=}0.04$). Every episode is faster than stepwise fallback, and the trusted oracle accounts for $99.99\%$ of online time.

\begin{table}[t]
\centering
\caption{Second UC instance ($100$ generators, $n{=}50$). The certified hold-drafter acceleration replicates at zero violations.}
\label{tab:supp_uc_second}
\scriptsize
\setlength{\tabcolsep}{5pt}
\begin{tabular*}{\columnwidth}{@{\extracolsep{\fill}}lr@{}}
\toprule
Metric & Value \\
\midrule
Wall speedup [CI] & $2.09\times\,[2.03,2.17]$ \\
Call reduction [CI] & $67.8\%\,[67.2,68.3]$ \\
Mean regret [CI] & $0.0136\,[0.0129,0.0143]$ \\
p95 / max regret & $0.0181$ / $0.0206$ \\
Mean accepted prefix & $11.12$ \\
Violations & $0$ \\
\bottomrule
\end{tabular*}
\end{table}

\section{Certificate Audit Components}
This section reports the audit quantities used in Sections 3.3 and 3.5. The ledger counts accepted steps separately from error-charged accepted segments.

\begin{table*}[t]
\centering
\caption{EMS certificate audit components. The band gate reduces realized excess while maintaining complete segment coverage.}
\label{tab:supp_cert_ems}
\scriptsize
\setlength{\tabcolsep}{3pt}
\begin{tabular*}{\textwidth}{@{\extracolsep{\fill}}lrrrrrrrrrr@{}}
\toprule
Row & $n$ & $|\mathcal{A}|$ & $\sum k_i$ & Mean $k$ & $\tau$ budget & Error & Envelope & Realized & Bound/real. & Coverage \\
\midrule
Learned boundary & $30$ & $5.13$ & $14.53$ & $2.83$ & $0.796$ & $0.000$ & $0.796$ & $0.490$ & $1.62\times$ & $93.3\%$ \\
Learned band gate & $30$ & $3.20$ & $9.00$ & $2.81$ & $0.591$ & $0.320$ & $0.910$ & $0.068$ & $13.45\times$ & $100.0\%$ \\
\bottomrule
\end{tabular*}
\end{table*}

\begin{table*}[t]
\centering
\caption{UC n=50 certificate audit envelope. The hold and Qwen rows keep zero violations while the accepted-step ledger substantially exceeds the number of error-charged segments. The $\tau$ coverage column is the fraction of accepted segments within the nominal $\tau$-budget; the conformally calibrated per-segment error charge (p95 error) prices the small residual, so realized regret stays bounded.}
\label{tab:supp_cert_uc}
\scriptsize
\setlength{\tabcolsep}{3pt}
\begin{tabular*}{\textwidth}{@{\extracolsep{\fill}}lrrrrrrrrr@{}}
\toprule
Row & $n$ & $\tau$ & Segments & Steps & $\tau$ coverage & Mean reg. & p95 reg. & Max reg. & p95 error \\
\midrule
Hold-commitment CGPA & $50$ & $0.04$ & $102$ & $1526$ & $100.0\%$ & $0.0145$ & $0.0173$ & $0.0179$ & $0.0066$ \\
Qwen3-8B CGPA & $50$ & $0.15$ & $316$ & $1543$ & $98.0\%$ & $0.0211$ & $0.0712$ & $0.1552$ & $0.0041$ \\
\bottomrule
\end{tabular*}
\end{table*}

\section{LLM Proposal Protocol}
This section records the proposal protocol used by the frozen LLMs. The prompt template, parser, and deterministic decoding settings are shared across models; the verifier and fallback contract remain unchanged across proposal sources.

\paragraph{EMS action proposal.}
The EMS prompt presents the current hour, the battery state of charge, the
running peak import, the allowed discrete action set, the battery and grid
feasibility limits, and the $K$-step load/PV/price forecast, and asks for one
discrete battery action per hour (positive discharges, negative charges),
returned as a JSON list of $K$ actions drawn from the allowed set.

\paragraph{UC commitment-change proposal.}
The UC prompt presents the current commitment vector, the per-generator minimum
and maximum output, the marginal costs, the units currently locked on or off by
minimum up/down times, and the $K$-step demand forecast, and asks --- for each
of the next $K$ steps --- for the set of generators to toggle relative to the
previous step, returned as a JSON list of $K$ index sets. It requests the
smallest commitment change that keeps each step's committed capacity able to
meet demand without toggling locked units.

\begin{table}[t]
\centering
\caption{LLM decoding and parsing protocol. The same settings are used across Qwen3-1.7B/4B/8B and Gemma-4-E2B/E4B/12B.}
\label{tab:supp_llm_protocol}
\scriptsize
\setlength{\tabcolsep}{4pt}
\begin{tabular}{lL{0.58\linewidth}}
\toprule
Item & Protocol \\
\midrule
Decoding & Greedy decoding with temperature $0$, top-$p=1$, and a fixed $1500$-token generation cap. \\
Chat format & The model's standard chat template, with reasoning output disabled where supported. \\
JSON extraction & A deterministic parser reads the first JSON object and extracts its proposed action list (EMS) or commitment-change list (UC). \\
Invalid output & Unparseable responses defer immediately to the trusted oracle. \\
Safety owner & The LLM never applies actions directly; every proposed prefix is transition-checked and value-checked before execution. \\
\bottomrule
\end{tabular}
\end{table}

\section{UC Timing Breakdown}
This section reports the per-decision timing breakdown referenced in Section 3.5. The trusted MILP oracle dominates online time; drafting, verification, and boundary inference together consume at most $2.04\%$ of per-decision time (Table~\ref{tab:supp_timing}). For the warm-start comparison referenced in Section 3.5, warm-starting the MILP with the same proposal (instead of skipping it) reaches only $0.91\times$ end-to-end speedup on the same $n{=}50$ UC run set: it still pays the full per-step solve, so warm-starting cannot substitute for skipping the oracle.

\begin{table*}[t]
\centering
\caption{UC n=50 per-decision timing. Shares are computed from measured component times over the formal n=50 runs.}
\label{tab:supp_timing}
\scriptsize
\setlength{\tabcolsep}{4pt}
\begin{tabular*}{\textwidth}{@{\extracolsep{\fill}}lrrrrrr@{}}
\toprule
Row & Oracle ms & Draft ms & Verify ms & Boundary ms & Oracle share & Other share \\
\midrule
Hold-commitment CGPA & $27838.5$ & $0.027$ & $0.432$ & $3.976$ & $99.98\%$ & $0.02\%$ \\
Qwen3-8B CGPA & $11387.2$ & $234.446$ & $0.269$ & $2.965$ & $97.96\%$ & $2.04\%$ \\
\bottomrule
\end{tabular*}
\end{table*}

\section{Verifier-Error Robustness}
This section reports the verifier-robustness sweep referenced in Section 3.2 (EMS, $n{=}100$, $\tau{=}0.04$). The verifier screens prefixes against a next-state estimate corrupted by bounded uniform noise (a fraction of the state range), while the executed trajectory uses the true dynamics; applied violations are measured on the true trajectory. Across the full noise range, no corrupted prefix is ever admitted (false-accept count zero at every level): violations stay at zero and the verifier instead defers more, so error is absorbed as reduced amortization rather than reduced safety.

\begin{table}[t]
\centering
\caption{Verifier-error robustness on EMS ($n{=}100$, a separate verifier-perturbation harness; the noise${=}0$ row is this harness's own reference). Noise is the bounded uniform error injected into the verifier's next-state estimate, as a fraction of the state range. Violations stay zero throughout; call reduction degrades gracefully.}
\label{tab:supp_verifier}
\scriptsize
\setlength{\tabcolsep}{4pt}
\begin{tabular*}{\columnwidth}{@{\extracolsep{\fill}}lrrr@{}}
\toprule
Verifier noise & Call red. & Mean reg. & Viol. \\
\midrule
$0.0$  & $63.5\%$ & $0.0077$ & $0$ \\
$0.005$ & $63.0\%$ & $0.0081$ & $0$ \\
$0.01$ & $62.8\%$ & $0.0122$ & $0$ \\
$0.02$ & $60.2\%$ & $0.0129$ & $0$ \\
$0.05$ & $55.7\%$ & $0.0159$ & $0$ \\
$0.10$ & $51.5\%$ & $0.0200$ & $0$ \\
$0.20$ & $44.5\%$ & $0.0188$ & $0$ \\
\bottomrule
\end{tabular*}
\end{table}

\section{White-Box Adaptive Adversary}
We stress the proposal slot with white-box adaptive adversaries that re-plan at every decision with full knowledge of the verifier, value boundary, dynamics, stage cost, and fallback oracle; they may craft any proposal but cannot modify the trusted core. One adversary maximizes true constraint-violation severity; a second maximizes the realized cost of accepted prefixes. Both are evaluated on EMS ($n{=}100$) and deployment-scale UC ($n{=}50$) under the same boundary, exact verifier, and fallback as the main runs.

The violation-maximizing adversary drives raw (un-wrapped) execution to a constraint breach in \emph{every} episode on both domains (Direct), yet CGPA admits none (Table~\ref{tab:supp_adversary}): it defers at every step, so amortization collapses to the trusted oracle floor (call reduction $0$) while applied violations stay at zero. A cost-maximizing adversary that keeps proposals feasible likewise causes zero applied violations. An adversary can erase CGPA's amortization but cannot breach its safety, because the feasibility gate never queries the value boundary the adversary attacks.

\begin{table}[t]
\centering
\caption{White-box adaptive violation-maximizing adversary: per-episode applied violation rate (raw execution versus CGPA) and CGPA call reduction. The adversary re-plans each step to breach constraints; CGPA holds applied violations at zero and absorbs the attack as lost amortization (call reduction falls to the oracle floor).}
\label{tab:supp_adversary}
\small
\begin{tabular*}{\columnwidth}{@{\extracolsep{\fill}}lrrr@{}}
\toprule
Domain & Direct viol. & CGPA viol. & CGPA call red. \\
\midrule
EMS ($n{=}100$) & $1.00$ & $0$ & $0\%$ \\
UC ($n{=}50$)   & $1.00$ & $0$ & $0\%$ \\
\bottomrule
\end{tabular*}
\end{table}

\paragraph{Cost degrades with attack strength; safety does not.} The
cost-maximizing adversaries refine this picture. On UC a cost-budgeted
adversary --- bounded to local cost envelopes --- leaves the regret certificate
fully covered (coverage $1.0$) and, within this adversary harness, attains both
lower regret and higher call reduction than the hold heuristic ($0.014$ vs
$0.026$ mean regret, $98\%$ vs $93\%$ call reduction; paired $95\%$ CIs exclude
zero), while an unbounded cost-maximizer
exhausts the nominal envelope (coverage $0$) at still-zero violations. On EMS
the bounded adversary matches the forecast drafter's own envelope
(matching its regret) at zero violations. Across attack strengths, applied
violations never leave zero. An adversary buys reduced
amortization or paid regret, never an unsafe transition.

\paragraph{A learned adversary fares no better.} We also train a \emph{learned}
white-box attacker --- an MLP that optimizes feasible high-cost prefixes directly
against the frozen contract ($n{=}100$) --- in place of the hand-crafted search.
Its feasible (un-wrapped) mean regret reaches $0.570$, the highest of any feasible
proposal source we evaluate and above the forecast drafter ($0.439$). The same
frozen contract still holds it to zero applied violations at certificate coverage
$0.99$ and mean/p95 regret $0.014/0.022$, as the value guard defers its aggressive
prefixes (call reduction $33\%$). A learned attacker can raise feasible cost but
cannot convert proposal control into an unsafe transition.

\section{Out-of-Distribution Day Expansion and Win Rate Against the Stepwise Oracle}
This section reports the day-level statistics behind the distribution-shift result of Section 3.3 ($18$ held-out calendar-shifted days $\times$ $5$ learned-drafter seeds, $90$ paired day--seed runs). A win counts a run whose total cost is below the stepwise oracle reference on the same day. Mean regret is $-0.137$ $[-0.166,-0.110]$ for the learned source and $-0.082$ $[-0.121,-0.044]$ for the forecast source; the in-distribution learned-guard win rate is $0.50$.

\begin{table}[t]
\centering
\caption{CGPA win rates against the stepwise oracle under calendar shift. Wilson $95\%$ intervals; zero violations in every run.}
\label{tab:supp_winrate}
\scriptsize
\setlength{\tabcolsep}{4pt}
\begin{tabular*}{\columnwidth}{@{\extracolsep{\fill}}lrrrr@{}}
\toprule
Scope & Wins & $n$ & Rate & Wilson $95\%$ \\
\midrule
Learned, day level & $15$ & $18$ & $0.83$ & $[0.61,0.94]$ \\
Forecast, day level & $12$ & $18$ & $0.67$ & $[0.44,0.84]$ \\
Learned, day--seed level & $72$ & $90$ & $0.80$ & $[0.71,0.87]$ \\
Forecast, day--seed level & $60$ & $90$ & $0.67$ & $[0.56,0.76]$ \\
\bottomrule
\end{tabular*}
\end{table}

\section{Full Proofs}
This section expands the proof sketches of the main paper. Statement numbering refers to the main paper; $V=V^{\pi_F}$ is the oracle cost-to-go with terminal value $V(s_T)=0$, so $V(s_0)=J_F(s_0)$ by definition.

\paragraph{Lemma 1 (Per-segment value decomposition).}
An accepted segment $i$ is admitted by the exact guard only if $g_i + V(s_{i+1}) \le V(s_i) + \tau_i|V(s_i)|$; rearranging gives $g_i \le V(s_i)-V(s_{i+1})+\tau_i|V(s_i)|$. For a repair segment $j$, define the realized repair slack $\rho_j := \max\{0,\, f_j + V(s_{j+1}) - V(s_j)\} \ge 0$; then $f_j \le V(s_j)-V(s_{j+1})+\rho_j$ holds by construction. If the repair applies $\pi_F$ with the same remaining horizon, the Bellman recursion $V(s_j)=f_j+V(s_{j+1})$ gives $\rho_j=0$. \hfill$\square$

\paragraph{Certificate 1 (Exact-boundary bound).}
The rollout partitions into consecutive segments whose end state is the next start state: $s_0=z_0,z_1,\dots,z_M=s_T$. Summing the Lemma~1 inequalities over all segments,
\begin{align*}
J_{\mathrm{CGPA}}(s_0)&=\sum_{i\in\mathcal A} g_i+\sum_{j\in\mathcal R} f_j\\
&\le \sum_{m=0}^{M-1}\bigl(V(z_m)-V(z_{m+1})\bigr)\\
&\quad +\sum_{i\in\mathcal A}\tau_i|V(s_i)|+\sum_{j\in\mathcal R}\rho_j .
\end{align*}
The first sum telescopes to $V(s_0)-V(s_T)=V(s_0)=J_F(s_0)$. Subtracting $J_F(s_0)$ yields the bound; dividing by $\max(|J_F(s_0)|,10^{-12})$ recovers the normalized regret of Eq.~(2). \hfill$\square$

\paragraph{Certificate 2 (Approximate-boundary bound).}
An accepted segment satisfies the deployed rule $g_i+\hat V(s_{i+1})\le \hat V(s_i)+\tau_i|\hat V(s_i)|$, so $g_i \le \hat V(s_i)-\hat V(s_{i+1})+\tau_i|\hat V(s_i)|$. Assumption~2 at the two endpoints gives $\hat V(s_i)\le V(s_i)+\epsilon$, $-\hat V(s_{i+1})\le -V(s_{i+1})+\epsilon$, and $|\hat V(s_i)|\le |V(s_i)|+\epsilon$, hence
\begin{align*}
g_i &\le V(s_i)-V(s_{i+1})+\tau_i|V(s_i)|+(2+\tau_i)\epsilon\\
&\le V(s_i)-V(s_{i+1})+\tau_i|V(s_i)|+(2+\tau_{\max})\epsilon .
\end{align*}
Repair segments are unchanged from Lemma~1, and the telescoping argument of Certificate~1 applies verbatim; the boundary error enters once per accepted segment, contributing $(2+\tau_{\max})\epsilon|\mathcal A|$ in total. \hfill$\square$

\paragraph{Proposition 4 (Conformal forward regret bound).}
A rollout has at most $T{+}1$ segment-boundary states. Split-conformal validity under the stated exchangeability gives, for each boundary state $s$, $\Pr[\,|\hat V_{0.5}(s)-V(s)|\le\epsilon(s)\,]\ge 1-\delta/(T{+}1)$ when the multiplier $\alpha$ is the conformal quantile of the calibration scores at miscoverage $\delta/(T{+}1)$. A union bound over the boundary states makes the inequalities hold jointly with probability at least $1-\delta$. On this event, the Certificate~2 substitution applies with the endpoint-specific bands:
\begin{align*}
g_i &\le V(s_i)-V(s_{i+1})+\tau_i|V(s_i)|\\
&\quad +(1+\tau_i)\epsilon(s_i)+\epsilon(s_{i+1})\\
&\le V(s_i)-V(s_{i+1})+\tau_i|V(s_i)|+(2+\tau_{\max})\bar\epsilon_i ,
\end{align*}
with $\bar\epsilon_i=\max\{\epsilon(s_i),\epsilon(s_{i+1})\}$. Summing and telescoping as in Certificate~1 completes the proof. \hfill$\square$

\section{Non-Learned Value Boundaries and Unguarded Shields}
This section reports the baseline-boundary comparison referenced in Section 3.3, run on a common $n{=}100$ EMS episode draw across the four boundaries ($\tau{=}0.04$). The unguarded shield executes the verified proposal at every draft length and pays its full regret --- the executed trajectories are identical across $K\in\{1,2,4,8\}$, so the result is insensitive to draft length. The learned boundary answers per-state queries in milliseconds without online construction.

\begin{table}[t]
\centering
\caption{Value-boundary baselines on a common EMS episode set ($n{=}100$, a separate draw from Table~\ref{tab:supp_mechanism}; the four boundaries share this draw to isolate the boundary effect). The learned boundary attains the lowest mean regret and the best tail among deployable boundaries.}
\label{tab:supp_nonlearned}
\scriptsize
\setlength{\tabcolsep}{1pt}
\begin{tabular*}{\columnwidth}{@{\extracolsep{\fill}}lrrrr@{}}
\toprule
Boundary & Mean reg. [CI] & p95 & Call red. & Viol. \\
\midrule
Unguarded shield ($K{=}1$--$8$) & $0.4450\,[0.4319,0.4579]$ & $0.5817$ & $100.0\%$ & $0$ \\
Forecast value table & $0.0153\,[0.0086,0.0223]$ & $0.0804$ & $61.7\%$ & $0$ \\
Oracle-rollout value table & $0.0261\,[0.0174,0.0349]$ & $0.1124$ & $64.3\%$ & $0$ \\
Learned boundary & $0.0048\,[-0.0006,0.0110]$ & $0.0710$ & $63.1\%$ & $0$ \\
\bottomrule
\end{tabular*}
\end{table}

\section{Conceptual Positioning}
Table~\ref{tab:positioning} positions CGPA against the neighbouring lines of
work discussed in the main paper's Related Work. CGPA is closest to
speculative decoding in invocation pattern and to safety filters in the
trusted-runtime contract, but its accepted object is a verified multi-step
control prefix gated by a value boundary, and the proposal source is arbitrary.

\begin{table}[t]
\centering
\caption{Conceptual positioning of CGPA against neighbouring lines of work.}
\label{tab:positioning}
\scriptsize
\setlength{\tabcolsep}{3pt}
\begin{tabular}{L{0.17\linewidth}L{0.19\linewidth}L{0.19\linewidth}L{0.28\linewidth}}
\toprule
Line of work & Verified object & Trusted work timing & Difference from CGPA \\
\midrule
Speculative decoding & token prefix distribution & target model verifies a draft batch & no transition constraints or fallback-relative cost certificate \\
Selective prediction / learning to defer & per-instance abstention & defer to a trusted expert on low confidence & single-step abstention, not a value-gated multi-step executed prefix \\
Shields and safety layers & one proposed action & correction at unsafe or risky actions & safety correction without committing a value-gated future prefix \\
Predictive safety filters / CBFs & candidate control input & online projection or constrained solve & can serve as fallback/verifier; do not amortize accepted prefixes \\
Event-triggered control & controller update condition & recompute when trigger fires & state-error reuse rather than prefix-cost-bounded multi-step commit \\
Move-blocking MPC and rollout & optimized or evaluated lookahead & solver or rollout remains central & reduce variables or improve actions, not trusted invocations \\
CGPA & verified action prefix from an arbitrary source & solver only on rejection or repair & exact safety plus a cheap value gate and an auditable per-segment regret certificate \\
\bottomrule
\end{tabular}
\end{table}

\end{document}